\documentclass[final]{siamltex}
\usepackage{amsmath,amssymb}
\usepackage{graphicx}        
\usepackage{color}

\newtheorem{remark}{Remark}
\newtheorem{example}{Example}[section]

\title{Saliency Based Control in Random Feature Networks
\thanks{This work was supported by the U.S. Office of Naval Research under MURI Grant Number N00014-10-1-0952.}}

\author{John Baillieul and Zhaodan Kong\thanks{Mar 14, 2014.}}

\begin{document}

\title{Saliency Based Control in Random Feature Networks}
\author{John Baillieul and Zhaodan Kong}
\date{\vspace{-11pt}}

\maketitle 
 \let\thefootnote\relax\footnotetext{John Baillieul and Zhaodan Kong are with the Department of Mechanical Engineering; John Baillieul is also with the Department of Electrical and Computer Engineering and the Division of Systems Engineering at Boston University, Boston, MA 02115. Corresponding author is John Baillieul (Email: johnb@bu.edu). \newline This work was supported by the U.S. Office of Naval Research under MURI Grant Number N00014-10-1-0952.}

\begin{abstract}
The ability to rapidly focus attention and react to salient environmental  features enables animals to move agiley through their habitats.  To replicate this kind of high-performance control of movement in synthetic systems, we propose a new approach to feedback control that bases control actions on randomly perceived features.  Connections will be made with recent work incorporating communication protocols into networked control systems.  The concepts of {\em random channel controllability} and {\em random channel observability} for LTI control systems are introduced and studied.
\end{abstract}


\section{Introduction}



Very agile motions of athletes and many species of animals depend on high-speed reactions to perceptions of selected features in the environment.  As pointed out in the seminal paper of Linkser (\cite{Linsker}), discovering the principles that underlie perceptual processing is important both for neuroscience and for the development of synthetic perceptual systems.  In what follows, we develop some simple feedback control laws that model reactions to salient elements of rapidly evolving feature networks.  The models make contact with very recent work on networked control systems where sensors and actuators are separated by a communication network of limited capacity. (\cite{Hristu},\cite{Yu}).

Much of the motivation for the work reported below comes from attempts to understand the ways in which animals navigate based on visual perception of environmental features.  More specifically, our previous work in \cite{Kong} and \cite{Sebesta} discusses ways in which a flying animal's perceptions of optical flow can provide cues that enable goal-directed steering through fields of obstacles.  Several steering laws based on optical flows generated both by single features and pairs of features were proposed. The argued biological plausibility of the laws is supported by simulations that reconstruct actual animal flight paths by incorporating appropriate protocols for switching among the laws.  

While the results were encouraging, it is clear that further research is needed in several directions to more completely understand bio-mimetic motion control algorithms.  One limitation of the laws proposed in \cite{Kong} and \cite{Sebesta} is that they utilize only one or two features at a time.  Animals undoubtedly are sensitive and reacting to larger numbers of features as they fly.  It was also the case that the features on which flight simulations were based were chosen somewhat mechanistically. Biomimetic response to optical flow undoubtedly involves dynamic and continual refocusing of attention based on visual saliency (\cite{Tsotsos}, \cite{Weichselgartner}).  Our current research continues to be focused on ephemeral sensed data of the type that is generated by sparse optical flow (\cite{Bru-Wei.2005}). At the same time, the work described in what follows is aimed at developing a conceptual understanding of feedback control that is based on streams of data coming from random elements in a sensor array (e.g.\ photo-receptors in the eye) with each sensor providing data over time intervals that may be quite short relative to the time-constants of the system being controlled.  To ground this work in a familiar context, we focus on finite dimensional time invariant linear systems of the type that are well known to anyone with a beginning graduate course in control theory.

\section{Problem Formulation}

Consider a discrete-time LTI system with $m$ inputs and $q$ outputs, whose evolution and output are given by
\begin{equation}
\begin{array}{rcl}
x(k+1)&=&Ax(k)+Bu(k),\ \ \ x\in\mathbb{R}^n,\ u\in\mathbb{R}^m,\\[0.08in]
y(k)&=&Cx(k), \ \ \ y\in\mathbb{R}^q.
\end{array}
\label{eq:jb:basic}
\end{equation}
Assume that $(A,B)$ is a controllable pair, and $(A,C)$ is an observable pair.  As in \cite{Hristu} and \cite{Yu}. we shall be interested in the evolution and output of (\ref{eq:jb:basic}) when only a portion of the input and output channels are active at any given step.  Hence, we are explicitly assuming that $m>1$, $q>1$.  Unlike these earlier references, however, we shall consider the case that the input and output channels are available randomly.  This rapidly changing random availability of input and output channels will capture an essential aspect of sensory perception and reaction to feature networks in biological and other systems operating in dynamic environments.

To fix ideas, we begin by considering systems with randomly available input channels.  The problem will be to design control inputs that work to steer the system (\ref{eq:jb:basic}) no matter what input channel sequence is encountered.  The following two problem setups are considered:
\smallskip
\begin{itemize}
\item{\bf Problem 1:}  Each of the $m$ input channels has probability $p$ of being active at each time step (and a complementary probability $1-p$ of being unavailable).  At each time step, there is an {\em activity pattern} represented by a vector of $m$ $0$'s and $1$'s.  The vector has a $1$ in the $i$-th place if the $i$-th input channel is active, and a $0$ in the $i$-th place if the channel is inactive.  Any pattern with $k$ $1$'s has probability of $p^k(1-p)^{m-k}$ of being observed.  There are $n\choose k$ such patterns, and assuming the probabilities are independent, the probability of observing exactly $k$ active input channels is ${n\choose k}p^k(1-p)^{n-k}$.
\item{\bf Problem 2:} Only one of the $m$ channels is active at any given time, and the $i$-th channel has probability $p_i$ of being active ($\sum p_i=1$.)  Label the columns of $B$ as $b_1,\dots,b_m$.  Under the protocol of random selection of the input channels, let $\gamma(k)$ denote the channel that is selected at the $k$-th step.  Letting the system (\ref{eq:jb:basic}) operate for $k_f$ steps, the state that is reached from $x(0)=0$ will be
\begin{equation}
\label{eq:jb:sequence}
\begin{array}{ll}
&b_{\gamma({k_f})}u(k_f)+Ab_{\gamma({k_f-1})}u(k_f-1)\\
&\cdots +A^{k_f-1}b_{\gamma(0)}u(0).
\end{array}
\end{equation}
\end{itemize}
Because of space limitations, only the setup in Problem 2 will be considered.  Given the assumption that $(A,B)$ is a controllable pair, we expect that if $k_f$ is sufficiently large, the set
\[
\{ b_{\gamma({k_f})},Ab_{\gamma({k_f-1})},\dots,A^{k_f-1}b_{\gamma(0)} \}
\]
will span $\mathbb{R}^n$.  The following will provide a useful background for proving a general statement.

\begin{lemma}
Let ${\bf X}=\{x_1,\dots,x_n\}\subset\mathbb{R}^n$ be a set of $n$ linearly independent vectors.  Let $\{\gamma(j)\}_{j=0}^{k-1}$ be a random sequence of integers drawn from the uniform distribution on the set $\{1,2,\dots,n\}$.  The probability that the corresponding sequence $\{x_{\gamma(0)},x_{\gamma(1)},\dots,x_{\gamma(k-1)}\}$ spans $\mathbb{R}^n$ is
\[
p(n,k)=\frac{n!S(k,n)}{n^k},
\]
where $S(k,n)$ is the Stirling number of the second kind denoting the number of ways to partition a set of $k$ objects into $n$ nonempty subsets.
\end{lemma}

\begin{proof}
For the case $k<n$, $S(k,n)=0$, which is of course consistent with the fact that there is no spanning set with fewer than $n$ elements.  Consider the set of all sequences of length $k\ge n$ such that each of the $n$ elements of {\bf X} appears one or more time in the corresponding vector sequence.  To determine the number of ways of assigning the members of {\bf X} to positions in a sequence, we count the number of partitions of sequence positions into $n$ subsets.  This is given by $S(k,n)$.  For each partition, there are $n!$ ways to assign the $n$ elements of {\bf X} to the $n$ partition cells.  Thus, among the total of $n^k$ possible vector sequences of length $k$, $n!S(k,n)$ will span $\mathbb{R}^n$.
\end{proof}

\begin{remark} 
Given well-known properties of Striling numbers of the second kind (\cite{Abramowitz}), it is possible to write $p(n,k)$ explicitly as
\begin{equation*}
\begin{array}{ll}
p(n,k)&=1-n(\frac{n-1}{n})^k+\frac{n(n-1)}{2}(\frac{n-2}{n})^k\\
&+{o}[(\frac{n-2}{n})^k].
\end{array}
\end{equation*}
For any $n$, the probability of a random vector sequence of the type described spanning $\mathbb{R}^n$ approaches $1$ as $k\to\infty$, but for large $n$ this approach can be slow.  (See Fig.\ \ref{fig:jb:SpanProb}.)
\end{remark} 

\begin{figure}[h]
\begin{center}
\includegraphics[width=\columnwidth]{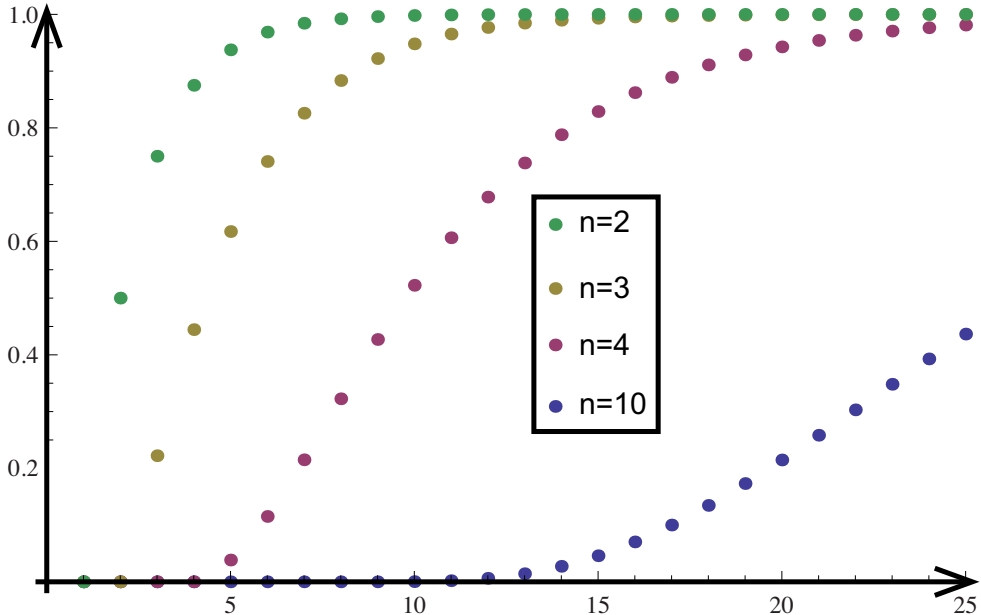}
\end{center}
\caption{For a spanning set {\bf X} of vectors in $\mathbb{R}^n$ for any n, the probability that a random sample of size $k\ge n$ will span approaches $1$ as $k$ becomes large.  The rate of approach decreases for large $n$, as illustrated in the figure (for $n=2,3,4,10$).}
\label{fig:jb:SpanProb}       
\end{figure}

\begin{proposition}
Suppose a random sequence of vectors is created by sequentially drawing elements from {\bf X} \underbar{with replacement}.  The mean number of draws before creating a spanning sequence is
\begin{equation}
{\cal M}_n= \sum _{k=n}^{\infty } k \left(1-\frac{n!}{ n^{k}} S(k,n)\right).
\label{eq:jb:MeanNoSpan}
\end{equation}
\end{proposition}
\begin{proof}
From Lemma 2.1 it follows that the probability of a sequence of length $k\ge n$ not spanning $\mathbb{R}^n$ is $ \left(1-\frac{n!}{ n^{k}} S(k,n)\right)$.   The result follows from the definition of mean length of non-spanning sequences of length greater than or equal to $n$..
\end{proof}

\begin{remark} 
We are not aware of any simpler expression for the mean non-spanning sequence length (\ref{eq:jb:MeanNoSpan}).  From Fig. \ref{fig:jb:SpanProb}, and the above remarks, we expect that the mean number of draws before obtaining a spanning sequence will increase with $n$.  For specific values of $n$, (\ref{eq:jb:MeanNoSpan}) can be computed, and several values are given in Table 1.  Simple data analysis suggests that mean length of non-spanning sequences grows quadratically with $n$.
\end{remark} 

\bigskip

\noindent
\begin{footnotesize}
\begin{tabular}{|| c || c|c|c|c|c||}   \hline
$n$ & 2 & 3 & 4 & 5 & 6 \\ \hline 
${\cal M}_n$ &  3.0  & 14.75 & 36.7778 & 71.0486 & 119.19  \\ \hline
$n$ & 7 & 8 & 9 & 10&\\ \hline 
${\cal M}_n$ &182.6  & 262.511 & 360.024 & 476.141& \\  \hline
\end{tabular}
\end{footnotesize}
\smallskip
\begin{center}
{\bf Table 1}
\end{center}

\section{Random Channel Controllability}

Returning to consideration of (\ref{eq:jb:basic}) under input channel constraints, the following definition will be useful.  In this section, we follow \cite{Yu} and \cite{Hristu} and assume that the matrix $A$ in (\ref{eq:jb:basic}) is invertible.  This will enable the discussion  that follows, but in work that will be reported elsewhere, it will be shown,that the case of non-invertible $A$ matrices is of interest in studying random input/output channel access.

\begin{definition}
Let $(A,B)$ be a controllable pair with $A$ $n\times n$ and $B$ $n\times m$ as in (\ref{eq:jb:basic}) .  Denote the columns of $B$ by $b_1\dots,b_m$.  The system (\ref{eq:jb:basic}) is said to be {\em random channel controllable} if for any set of indices $\gamma(0),\dots,\gamma(n-1)$ in which every element of the set $\{1,\dots,m\}$ appears at least once, the corresponding set of vectors
\[
\{b_{\gamma(n-1)},Ab_{\gamma(n-2)},\dots,A^{n-1}b_{\gamma(0)}\}
\]
spans $\mathbb{R}^n$.
\end{definition}

\begin{example}
Suppose $m=n$.  Let
\[
A=\left(\begin{array}{cccc}
\lambda_1 &  0 & 0 & 0\\
0 & \lambda_2 & 0 & 0\\
\vdots & \vdots & \ddots & \vdots \\
0 & 0 & 0 & \lambda_n \end{array}\right)
\ \ B=I.
\]
The system decouples into $n$-scalar systems.  If $\lambda_j\ne 0$ for and $j=1\dots,n$, it is random channel controllable.
\end{example}

\begin{example}
The system
\[
A=\left(\begin{array}{ccc}
\lambda_1 & 0 & 0\\
0 & \lambda_2 & 0 \\
0 & 0 & \lambda_3\end{array}\right),\ \ B=\left(
\begin{array}{cc}
0 & 1\\
1 & 0\\
1 & 0\end{array}\right)
\]
is controllable if $\lambda_2\ne \lambda_3$ but not random channel controllable because
\[
\left(\begin{array}{ccc}
b_1 & Ab_2 & A^2b_2\end{array}\right) = \left(\begin{array}{ccc}
0 & \lambda_1 & \lambda_1^2\\
1 & 0& 0\\
1 & 0 & 0\end{array}\right)
\]
is singular for all choices of $\lambda_1$.
\end{example}

\begin{example}
Let
\[
A=\left(\begin{array}{ccc}
\lambda_1 & 0 & 0\\
0 & \lambda_2 & 0 \\
0 & 0 & \lambda_3\end{array}\right),\ \ B=\left(
\begin{array}{cc}
0 & 1\\
1 & 1\\
1 & 0\end{array}\right).
\]
The system \underbar{is} random channel controllable because all matrices of the form
\[
\left(\begin{array}{ccc}
b_i & Ab_j & A^2b_k\end{array}\right)
\]
are invertible provided it is not the case that $i=j=k$.
\end{example}

With these examples in mind, we state the following.

\begin{proposition}
Consider the system (\ref{eq:jb:basic}) under the assumption that $(A,B)$ is a controllable pair.  Let $\gamma(0),\dots,\gamma(k_f-1)$ be a random sequence drawn sequentially with equal probability and replacement from the index set $\{1,\dots,m\}$.  If $k_f\ge n$, the probability $p(n,k_f)$ that the corresponding random vector sequence
\[
b_{\gamma(k_f-1)},Ab_{\gamma(k_f-2)},\dots,A^{k_f-1}b_{\gamma(0)}
\]
spans $\mathbb{R}^n$ satisfies the following inequality
\begin{equation}
0<p(m,k)\le\frac{m!\,S(k,m)}{m^k},
\label{eq:jb:inequality}
\end{equation}
where $S(k,n)$ is the Stirling number of the second kind denoting the number of ways to partition a set of $k$ objects into $n$ nonempty subsets, and where we have simplified notation by writing $k=k_f$.  If the system (\ref{eq:jb:basic}) is random channel controllable, the right hand inequality holds with equality.  If the system is not random channel controllable, the right hand inequality is strict.
\end{proposition}

\begin{proof}
Consider the set of sequences of length $k\ge n$ having the form
\begin{equation}
b_{\gamma(k-1)},Ab_{\gamma(k-2)},\dots,A^{k-1}b_{\gamma(0)},
\label{eq:jb:randomVectors}
\end{equation}
where each of the $m$ channel indices appears one or more times in the index sequence $\gamma(0),\dots,\gamma(k-1)$.  As in Lemma 2.1, we found the number  of partitions of the set of sequence positions $\{1,\dots, k\}$ into $m$ subsets.  This is $S(k,m)$.  Each such partition  can have its cells labeled by the indices $1,\dots, m$ in $m!$ different ways.  Hence, among the total number of $m^k$ vector sequences of the form (\ref{eq:jb:randomVectors}), $m!S(k,m)$ span $\mathbb{R}^N$, and this shows that
\[
p(m,k)=\frac{m!\,S(k,m)}{m^k}
\]
in the case that (\ref{eq:jb:basic}) is random channel controllable.  When (\ref{eq:jb:basic}) is not random channel controllable, a smaller fraction of vector sequences of the form (\ref{eq:jb:randomVectors}) span $\mathbb{R}^n$, and hence the inequality (\ref{eq:jb:inequality}) becomes strict.
\end{proof}

{\bf Remark:} It might seem surprising that $p(m,k)$ in (\ref{eq:jb:inequality}) is an increasing function of the number $m$ of input channels, but this is not the case, since if there are more channels being randomly chosen, there is a greater likelihood that a non-spanning sequence may occur.  Indeed, in the trivial (mimimum number of channels) case, $m$=1, $p(m,k)\equiv 1$---reflecting the fact that any sequence of $n$ vectors in (\ref{eq:jb:randomVectors}) will span.

{\bf Remark:}  When there are $n$ input channels ($n=$ the dimension of the state), typical mean lengths of non-spanning sequences are given in Table 1.  If we fix the number of channels at $2$, these mean lengths of non-spanning sequence are considerably reduced.  For large values of $n$, the probability of a randomly selected set of vectors of the form (\ref{eq:jb:randomVectors}) spanning is high.

\section{Random Channel Observability}

There is of course a corresponding notion of {\em random channel observability}.  As in Section 3, we assume that $A$ in (\ref{eq:jb:basic}) is invertible.

\begin{definition}
Let $(A,C)$ be an observable pair with $A$ $n\times n$ and $C$ $q\times n$ as in (\ref{eq:jb:basic}) .  Denote the rows of $C$ by $c_1\dots,c_q$.  The system (\ref{eq:jb:basic}) is said to be {\em random channel observable} if for any set of indices $\gamma(0),\dots,\gamma(n-1)$ in which every element of the set $\{1,\dots,q\}$ appears at least once, the corresponding set of vectors
\[
\{c_{\gamma(n-1)}^T,A^Tc_{\gamma(n-2)}^T,\dots,{A^T}^{n-1}c_{\gamma(0)}^T\}
\]
spans $\mathbb{R}^n$.
\end{definition}

{\bf Remark:} As in the case of random channel controllability, not all observable pairs are random channel observable.  The following is dual to Proposition 3.2.

\begin{proposition}
Consider the system (\ref{eq:jb:basic}) under the assumption that $(A,C)$ is an observable pair.  Let $\gamma(0),\dots,\gamma(k_f-1)$ be a random sequence drawn sequentially with equal probability and replacement from the index set $\{1,\dots,q\}$.  If $k_f\ge n$, the probability $p(q,k_f)$ that the corresponding random vector sequence
\[
c_{\gamma(k_f-1)},A^Tc_{\gamma(k_f-2)}^T,\dots,{A^T}^{k_f-1}c_{\gamma(0)}^T
\]
spans $\mathbb{R}^n$ satisfies the following inequality
\begin{equation}
0<p(q,k)\le\frac{m!\,S(k,q)}{q^k},
\label{eq:jb:inequality2}
\end{equation}
where $S(k,q)$ is the Stirling number of the second kind denoting the number of ways to partition a set of $k$ objects into $q$ nonempty subsets, and where we have simplified notation by writing $k=k_f$.  If the system (\ref{eq:jb:basic}) is random channel controllable, the right hand inequality holds with equality.  If the system is not random channel controllable, the right hand inequality is strict.
\end{proposition}

\medskip

The proof is virtually the same as that of Proposition 3.2 and hence is omitted.

\section{Feedback stabilization with random input channel access}

The discussion in Section 3 is focused on conditions under which there exists a sequence of inputs that will steer a controllable system between prescribed initial and final states.  It does not prescribe an algorithm for sequentially choosing the inputs that achieve any particular control objective.  Following the work in \cite{Yu} and \cite{Hristu}, we consider the problem of stabilizing an LTI system with random input channel access of the type discussed above.  Because randomness is an essential feature of our problem, feedback stabilization can only be achieved in a probabilistic sense.  For systems with unstable modes, the controller's random access to input channels will make meaningful stabilization challenging, and feedback gains must be chosen to account for the probability that channels will be inaccessible for certain periods of time.  For the probabilistic access models under consideration, occasional randomly occurring inaccessibility of an input channel will lead to large excursions of the system dynamics.  It will be argued that such excursions will be statistically inevitable, and they will occur over time intervals with typical lengths that are similar to the time-to-failure intervals described in \cite{Nair}.

To fix ideas, consider Example 3.1.  It follows from Proposition 3.2 that when $n$ is large, it will require a relatively large number of input steps to produce a controlled motion in every dimension of the state space.  Indeed, when channels are selected sequentially and randomly with equal probability equal to $1/n$, the mean waiting time for any channel to be selected is exactly $n$.  This follows from simply realizing that the selection of each channel is a Bernoulli process with probability of selection $p=1/n$.  A further consequence of properties of mean waiting times in Bernoulli processes is that the variance of channel selections is $(1-p)/p^2 \sim n^2$.  This means that the mean trajectories of the scalar components $x_1(t),\dots,x_n(t)$ will not adequately characterize the dynamics of the system since with high probability there will be large waiting times between successive channel availabilities.

To further understand the complexity of this example, we consider a simple scalar process governed by a Bernoulli random switching.  Specifically suppose a real-valued process $x(\cdot)$ evolves according to

\begin{equation}
x(k+1)=\left\{\begin{array}{ccl}
ax(k)& with\  prob & p \\
bx(k) & with\  prob  & 1-p.\end{array}\right.
\label{eq:jb:singleChannel}
\end{equation}

\noindent The mean of this process is deterministic and at the $k-$th step has the value $(pa+(1-p)b)^kx(0)$.  The mean is thus asymptotically stable precisely when $|pa+(1-p)b|<1$.  The second moment of the process will be asymptotically stable when $|pa^2+(1-p)b^2|<1$, however, and from this, it is easy to see that there are processes that are stable in the mean but which have unbounded variance.  This may be problematic for systems depicted in Example 3.1 with unstable open-loop modes.  In these cases,  $a=\lambda_j$ and $p =  1-\frac1n$.  We can examine this more explicitly by considering the first and second moments of the process (\ref{eq:jb:singleChannel}).  The mean process is the simple linear equation $x(k+1)=\left[ pa+(1-p)b\right] x(k)$.  Restricting to non-oscillatory trajectories, the mean of the process will be asymptotically stable when $0<pa+(1-p)b<1$, which requires that $pa<1$.  For the case of Example 3.1, $p=(n-1)/n$ and the maximum size of an open-loop unstable mode, $\lambda_j$, can be written as a function of $n$: $a_n<n/(n-1)$.  

The second moment of the process (\ref{eq:jb:singleChannel}) will similarly be stable when $pa^2 + (1-p)b^2<1$.  Simple algebra again shows that the maximum size of a mode $\lambda_j$ leading to a stable second moment is $\lambda_j=a_n<\sqrt{(n/(n-1)}$. It is perhaps not surprising that both these bounds are similar to the stabilizability and error tolerance conditions associated with the data-rate theorem, \cite{Nair}.  Returning to the remark that Bernoulli processes have fat tailed distributions, we note that even when a feedback law stabilizes the means of the eigenmodes of a system of the form given in Example 3.1, the variance may be unbounded.  This is illustrated by the simulations plotted in Fig. \ref{fig:zk:simulation}.

\begin{figure}[h]
\begin{center}
\includegraphics[width=\columnwidth]{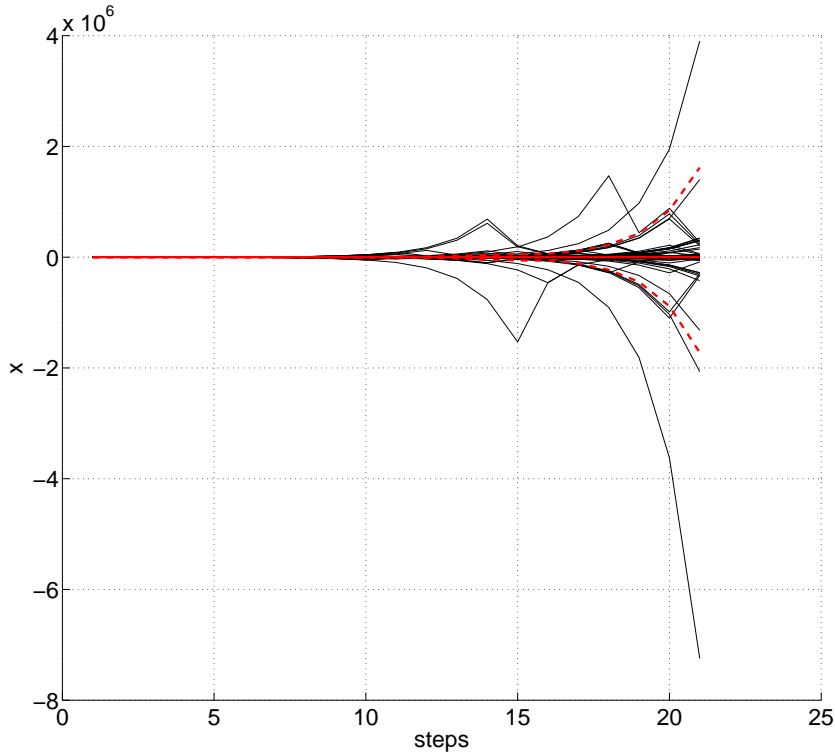}
\end{center}
\caption{The trajectories are from a system of the form of Example 3.1.  Here trajectories for a single mode are plotted, but the system has an $n$=3 dimensional state, all three open loop modes are $\lambda=3$, and the stabilizing feedbacks $k_1=k_2=k_3=-2.7$, rendering the closed loop modes equal to $0.3$.  The solid red curve is the mean of 1000 trajectories, the dotted red curves denote the empirical variance, and the black curves are 500 of the simulated trajectories---which clearly make major excursions from the mean. }
\label{fig:zk:simulation}       
\end{figure}

\section{Conclusions}

The paper has introduced problems of feedback control in which sensor and actuator connections to a plant are achieved  through randomly accessed communication channels.  The motivation for the discussion is biological navigation through feature networks (e.g. visually perceived features).  While this is most naturally related to random channel observability problems, limitations of space have led us to consider only random channel control problems.  The paper has introduced what are believe to be new notions that we have called {\em random channel controllability} and {\em random channel observability}.  These are stronger than the usual notions of controllability and observability in LTI systems.  These concepts suggest that we may think of a set of input (or output) channels as being salient if any $n$-step combination involving all of them is enough to steer the system to an arbitrary goal point.

The main results of the paper point to an increased likelihood of a random sequence of input channels providing controllability (observability) when the number of channels is not too large.  Specifically this is suggested by Propositions 3.2 and 4.2 as well as the example treated in Section 5.  It is partly  a consequence of the need to use all available channels to realize controllability combined with the restriction that only a single channel can be used at a time.  While we shall report some relaxation of the single channel restriction elsewhere, (and consider Problem 1 of the Introduction), we believe it is also worth observing that the limits on stabilization for multi-channel systems  may point to fundamental limits in the amount of information that can be processed to achieve desired control objectives.  Finally, we remark that future work will be aimed at developing probabilistic extensions of the output stabilizability results  in \cite{Yu} and \cite{Hristu} to the case of systems with random access input and output channels.



\section*{Acknowledgments}
The authors thank colleagues Xi Yu and Sean Andersson for the motivation provided by \cite{Yu}, and Kenn Sebesta for advocating visual sensing for controlling movement (\cite{Sebesta}). 


\end{document}